\documentclass[%
 reprint,
superscriptaddress,
nofootinbib,
 amsmath,amssymb,
 aps,
pra,
]{revtex4-1}

\usepackage{amsthm}
\usepackage{graphicx}
\usepackage{dcolumn}
\usepackage{bm}
\usepackage{float}
\usepackage{booktabs} 
\usepackage{amsmath,amsthm} 
\usepackage{setspace} 
\usepackage[sort&compress]{natbib}
\usepackage{extarrows}
\usepackage{multirow}

\allowdisplaybreaks[4] 

\usepackage{color}
\newcommand{\widesim}[2][1.5]{\mathrel{\overset{#2}{\scalebox{#1}[1]{$\sim$}}}}

\makeatletter

\newcommand{\Rmnum}[1]{\expandafter\@slowromancap\romannumeral #1@}
\makeatother

\begin{document}

\title{Local unitary classification for sets of generalized Bell states}

\author{Bujiao Wu}
\email{wubujiao@ict.ac.cn}
\author{Jiaqing Jiang}
\email{jiangjiaqing17@mails.ucas.ac.cn}
\author{Jialin Zhang}
\email{zhangjialin@ict.ac.cn}
\author{Guojing Tian}
\email{tianguojing@ict.ac.cn}
\author{Xiaoming Sun}
\email{sunxiaoming@ict.ac.cn}
\affiliation{
  CAS Key Lab of Network Data Science and Technology, Institute of Computing Technology, Chinese Academy of Sciences, 100190, Beijing, China.
}
\affiliation{
  University of Chinese Academy of Sciences, Beijing, 100049, China.
}

\begin{abstract}
In this paper, we study the local unitary classification for pairs (triples) of generalized Bell states, based on the local unitary equivalence of two sets. In detail, we firstly introduce some general unitary operators which give us more local unitary equivalent sets besides Clifford operators.
And then we present two necessary conditions for local unitary equivalent sets which can be used to examine the local inequivalence. Following this approach, we completely classify all of pairs in $d\otimes d$ quantum system into $\prod_{j=1}^{n} (k_{j} + 1) $ LU-inequivalent pairs when the prime factorization of $d=\prod_{j=1}^{n}p_j^{k_j}$. Moreover, all of triples in $p^\alpha\otimes p^\alpha$ quantum system for prime $p$ can be partitioned into $\frac{(\alpha + 3)}{6}p^{\alpha} + O(\alpha p^{\alpha-1})$ LU-inequivalent triples, especially, when $\alpha=2$ and $p>2$, there are exactly $\lfloor \frac{5}{6}p^{2}\rfloor + \lfloor \frac{p-2}{6}+(-1)^{\lfloor\frac{p}{3}\rfloor}\frac{p}{3}\rfloor + 3$ LU-inequivalent triples.
\begin{description}
\item[PACS numbers] 03.67.HK, 03.65.Ud.
\end{description}
\end{abstract}

\pacs{Valid PACS appear here}
\maketitle

\theoremstyle{remark}
\newtheorem{definition}{\indent{\bf Definition}}
\newtheorem*{observation}{\indent{\bf Observation}}
\newtheorem{lemma}{\indent{\bf Lemma}}
\newtheorem{theorem}{\indent{\bf Theorem}}
\newtheorem*{corollary}{\indent{\bf Corollary}}
\newtheorem*{conjecture}{\indent{\bf Conjecture}}

\newtheorem*{proposition}{\indent\em Proposition} 
\newtheorem{claim}{\indent{\bf\em Claim}}
\def\QEDclosed{\mbox{\rule[0pt]{1.3ex}{1.3ex}}}
\def\QED{\QEDclosed}
\def\proof{\indent{\bf\em Proof}.}
\def\endproof{\hspace*{\fill}~\QED\par\endtrivlist\unskip}

\section{\label{sec:level1}Introduction}

As is known, two local unitary (LU) equivalent quantum states play the same role in implementing quantum information processing tasks, and many fundamental properties including the maximal violations of Bell inequalities \cite{RH95,RFW01,MZ02,ML12}, the degree of entanglement \cite{WKW98,CHB96} and other quantum correlations \cite{HO01,WHZ03,LH01,Modi10,Roa2011} remain the same. Therefore, it has always been a very important research area finding out effective and efficient methods to give a complete LU-classification of all quantum states in the corresponding quantum system.

Actually, in many quantum information processing tasks, what we need is a set of quantum states rather than only an individual quantum state. Various sets of quantum states have been employed to design the corresponding quantum key distribution protocols \cite{terhal2001hiding,gea2002hiding,eggeling2002hiding,divincenzo2002quantum,luo2014multi,yang2007efficient,rahaman2015quantum,yang2015quantum}, especially quantum secret sharing protocols \cite{yang2007efficient,rahaman2015quantum,yang2015quantum}. In \cite{rahaman2015quantum}, the authors have presented a (2,n)-threshold quantum secret sharing protocol, where any two cooperating players from disjoint groups can always reconstruct the secret, based on the local discrimination of their specical GHZ-states set (the formal definition is in Section \ref{sec:level2}). Obviously, if we employ another LU-equivalent set to share this secret, we will get the same generality and efficiency as they claimed. And the only difference is that we have to operate the corresponding local unitaries on our set to derive their special set in \cite{rahaman2015quantum}. In this sense, the LU-equivalence of sets deserve much more consideration.

In addition, the local distinguishability, which has been widely studied in \cite{bennett1999quantum,ghosh2004distinguishability,walgate2000local,fan2004distinguishability,Owari06,hayashi2006,yu2012four,nathanson2013three,yu2015detecting}, will never change for two LU-equivalent sets of quantum states. That is to say, the LU-classification of all the sets of quantum states will help to assure the local distinguishability of the whole sets of the specific quantum system. There has already been a successful case in analyzing all the quadruples of generalized Bell states (GBSs) in $4 \otimes 4$ quantum system \cite{Tian16}. The authors have classified all the 1820 quadruples into 10 equivalent classes of LU-equivalent sets, and then discovered the fact that there are 3 locally indistinguishable quadruples and 7 locally distinguishable quadruples. As a result, the LU-equivalence of sets of quantum states deserves much more attention at least in the above circumstances.

However, very limited results have been obtained so far. Even in the work of Tian \textit{et al.}\cite{Tian16}, they considered the GBS-quadruples in one specific quantum system ($4\otimes 4$), and the GBS-triples in $p \otimes p$, $4 \otimes 4$ and $6 \otimes 6$ quantum systems. It is not hard to see that all these conclusions are about the nondegenerate or low-dimensional degenerate GBS-sets. In fact, the LU-equivalence of degenerate states is usually more complicated than that of nondegenerate states \cite{li2014local}. In this paper, by analyzing emphatically the properties of degenerate states, we obtain a complete classification of GBS-pairs in $d\otimes d$ quantum system for all positive integers $d$ and of GBS-triples in $p^{\alpha}\otimes p^{\alpha}$ quantum system for all prime $p$ and positive integer $\alpha$. Specifically, we consider generalized Pauli matrices (GPMs) since there is a one-to-one correspondence between GPMs and GBSs. Besides the Clifford operators~\cite{JM14}, we also present some more general unitary operators, to prove two GBS-sets are LU-equivalent. Moreover, we construct a new invariant which works efficiently to explain the local unitary inequivalence of two degenerate GBS-sets. The followings are our main results.
\begin{itemize}
\item[(a)] There are exactly $\prod_{j=1}^{n}(k_j+1)$ LU-inequivalent GBS-pairs in $d\otimes d$ quantum system when the prime factorization of $d=\prod_{j=1}^{n}p_j^{k_j}$.
\item[(b)] There are exactly $\lfloor \frac{5}{6}p^{2}\rfloor + \lfloor \frac{p-2}{6}+(-1)^{\lfloor\frac{p}{3}\rfloor}\frac{p}{3}\rfloor + 3$ LU-inequivalent GBS-triples in $p^{2}\otimes p^{2}$ quantum system for prime $p\geq 3$.
\item[(c)] There are $\frac{(\alpha + 3)}{6}p^{\alpha} + O(\alpha p^{\alpha-1})$ LU-inequivalent GBS-triples in $p^{\alpha}\otimes p^{\alpha}$ quantum system for all of prime $p$ and positive integer $\alpha$.
\end{itemize}
Our results show the LU-classifications for the individual MESs, the pairs and the triple are totally different. We wish our classification can serve for the classification of GBS-triples in all dimensions. The structure of the classification is much more complicated than we expected when the number of states in the set increases.

This paper is organized as follows. In Section \ref{sec:level2}, we review the definition of GBS and LU-equivalence. Next in Section \ref{sec:level3}, we introduce some unitary operators and build two necessary conditions for LU-equivalent GBS-triples. As applications of the above unitary operators and necessary conditions, we successfully classifies all of GBS-pairs in $d\otimes d$ quantum system for any dimension $d$ in Section \ref{sec:level4} , and we also give a classification of GBS-triples in $p^{\alpha}\otimes p^{\alpha}$ quantum system for all of primes $p$ in Section \ref{sec:level5}. Finally, we conclude this paper and emphasize some future work in Section \ref{sec:level6}.

\section{\label{sec:level2} Preliminaries}

We focus on the \emph{local unitary} (LU) equivalence of \emph{sets of generalized Bell states} (GBS-sets) in bipartite system $\mathcal{H}_{A}\otimes \mathcal{H}_{B}$ in this paper, where the dimensions of $\mathcal{H}_{A}$ and $\mathcal{H}_B$ are both $d$. For convenience, we denote the system as $d\otimes d$, and denote the computational basis of a single qudit as $\{|k\rangle|k\in\mathbb{Z}_d\}$. In the following we give a definition of GBS and LU-equivalence of GBS-sets.

Since GBS can be represented by \emph{generalized Pauli matrice} (GPM) (\cite{Tian16}), we firstly review the definition of GPM. Consider a $d\otimes d$ quantum system, the GPM is defined as
\begin{equation*}
U_{s,t}\triangleq X_{d}^{s}Z_{d}^{t}
\end{equation*}
where $s,t\in\mathbb{Z}_{d}=\{0,1,\cdots,d-1\}$, $X_{d}\triangleq\sum_{k\in\mathbb{Z}_{d}}|k + 1\rangle\langle k|$ is shift operator and $\quad Z_{d}\triangleq\sum_{k\in \mathbb{Z}_{d}}\omega^{k}|k\rangle\langle k|$ is clock operator with $\omega=\exp (2\pi i/d)$. We omit subscripts of $X_{d},Z_{d}$ when there is no ambiguity.
GBS is generated by operating a GPM locally on the \emph{standard maximally entangled state} (standard MES), that is,
\begin{equation*}
|\phi_{s,t}\rangle \triangleq (I\otimes U_{s,t})|\phi_{0,0}\rangle,
\end{equation*}
where $|\phi_{0,0}\rangle \triangleq \frac{1}{\sqrt{d}}\sum_{k\in\mathbb{Z}_{d}}|kk\rangle$ is the standard MES. Thus there is one-to-one correspondence between GBS and GPM. That is to say, some properties of GBS can be represented by those of the corresponding GPM, and the local unitary equivalence is one of them.

As is shown in Kraus's work \cite{Kraus10}, two bipartite states $|\phi\rangle,|\psi\rangle$ are called LU-equivalent, \emph{i.e.}, $|\phi\rangle\mapsto|\psi\rangle,$ if there exist local unitary operators $U_{A},U_{B}$ such that $|\phi\rangle = (U_{A}\otimes U_{B}) |\psi\rangle$ up to some global phase. Especially if $|\phi\rangle,|\psi\rangle$ are both GBSs, then the above LU-equivalence can be illustrated as the \emph{unitary equivalence} (U-equivalence) of their GPMs because the transpose operation keeps unitary, \emph{i.e.},
\begin{eqnarray*}
\begin{array}{ccc}
M &     =       & U_{B} N U_A^T \\
         & \triangleq  & U_L N U_R
\end{array}
\end{eqnarray*}
up to some global phase, where $M,N$ is the corresponding GPM of $|\phi\rangle,|\psi\rangle$ respectively, and $U_L\triangleq U_B, U_R\triangleq U_A^T$. We denote the U-equivalence of two GPMs as $M \sim N$, \emph{i.e.}, $M\approx U_{L} N U_{R} $, where ``$\approx$'' denotes ``equal up to some global phase''.

Next, we need to generalize the LU-equivalence of two GBSs to that of two GBS-sets, that is, we will define the U-equivalence of two GPM-sets based on that of two GPMs. Consider two GBS-sets $\{|\phi_1\rangle,\cdots, |\phi_n\rangle\}$ and $\{|\psi_1\rangle,\cdots, |\psi_n\rangle\}$, correspondingly their GPM-sets are denoted as $\mathcal{M}=\{M_{1},\cdots, M_{n}\}$ and $\mathcal{N}=\{N_{1},\cdots, N_{n}\}$. If there exist two unitary operators $U_A,U_B$ and a permutation $\pi$ over $\{1,\cdots,n\}$ such that $|\phi_i\rangle \approx (U_{A}\otimes U_{B}) |\psi_{\pi(i)}\rangle$ for each $i\in[n]$, then these two GBS-sets are called LU-equivalent, \emph{i.e.}, $\{|\phi_1\rangle,\cdots, |\phi_n\rangle\}\mapsto\{|\psi_1\rangle,\cdots, |\psi_n\rangle\}.$ Similarly, we derive the U-equivalence of the two GPM-sets, that is, if there exist $U_{L},U_{R},\pi$, and for any $i\in \{1,\cdots,n\}$, we have $U_L N_{{\pi(i)}} U_R\approx M_{i}$, then $\mathcal{M}$ and $\mathcal{N}$ are U equivalent, denoted as
\begin{equation*}
U_L\mathcal{N}U_R \approx \mathcal{M}.
\end{equation*}
Especially when $U_R= U^{\dagger}_L$, we call the two GPM-sets are \emph{unitary conjugate equivalent} (UC-equivalent), denoted as
\begin{equation*}
\mathcal{N} \stackrel{{U_L}}\sim \mathcal{M}.
\end{equation*}

In the following parts, we prefer employing the U-equivalence of GPM-sets to represent the LU-equivalence of GBS-sets because of the one-to-one correspondence. For simplicity, we will name the research subject, the GBS-sets with two or three elements, as GBS-pair and GBS-triple respectively. $\text{Inv}_{d}(k)$ is defined as an integer such that $k\cdot \text{Inv}_{d}(k) \equiv 1\text{ (mod } d)$ and $0<\text{Inv}_{d}(k)<d$. If there are no explicit explanation, $d$ is the dimension of quantum system, and we also use $k^{-1}$ to express the inverse of $k$ corresponding to $\mathbb{Z}_{d}$ for simplify. $k\perp d$ means that $k$ is co-prime to $d$. $a|d$ means that $a$ is a factor of $d$.


\section{\label{sec:level3}Conditions of LU-equivalence}

In this section, we will firstly review some useful unitary operators which can transform one GPM-set to another. Afterwards, we introduce two necessary conditions of U-equivalence between two GPM-sets.

\subsection{\label{sec:level3-1}Useful operators for unitary transformation}

Here we introduce Clifford operators and some more general unitary operators, to realize the U-equivalence of two GPM-sets. 

Clifford operators are unitary operators that map the Pauli group to itself under conjugation~\cite{Dehaene03}. For convenience of our classification, we introduce four common Clifford operators at first. Because of their simple transformation form, they also play an important role in the work of Tian et al.~\cite{Tian16}.

Operators $P,R$ are two basic Clifford operators. In fact, we can use them to generate all of the Clifford operators \cite{JM14}:
\begin{equation*}
\begin{aligned}
&P=\begin{cases}
\sum_{k\in\mathbb{Z}_{d}}{\omega}^{\frac{k(k-1)}{2}}|k\rangle\langle k|,\text{ if }d \text{ is odd.}\\
\sum_{k\in\mathbb{Z}_{d}}{\omega}^{\frac{k^{2}}{2}}|k\rangle\langle k|, \text{ if }d \text{ is even.}
\end{cases}\\
&R=\frac{1}{\sqrt{d}}\sum_{k,j\in\mathbb{Z}_{d}}\omega^{kj}|k\rangle\langle j|\\
\end{aligned}
\end{equation*}
The other two common Clifford operators are $V=P^{2}RPRP^{2}$ and $Q_{k}=RP^{k^{-1}}RP^{k}RP^{k^{-1}}$. Through direct computation, the above four Clifford operators realize the following UC-transformations
\begin{equation*}
\begin{aligned}
&X\stackrel{{P}}\sim XZ\text{ and }
Z\stackrel{{P}}\sim Z,\\
&X\stackrel{{R}}\sim Z\text{ and }
Z\stackrel{{R}}\sim X^{\dagger},\\
&X\stackrel{{V}}\sim X\text{ and }
Z\stackrel{{V}}\sim XZ,\\
&X \stackrel{{Q_k}}\sim X^{k^{-1}}\text{ and }
Z\stackrel{{Q_k}}\sim Z^{k} \text{ where } k\perp d.
\end{aligned}
\end{equation*}

These four UC-transformations are useful since they can transform two GPMs to the other two GPMs simultaneously.

Next, in the $p^{\alpha}\otimes p^{\alpha}$ system, we construct some other local unitary operators other than Clifford operators to help us move forward in the process of finding out more U-equivalent sets.

\begin{lemma}\label{permutation}
In a $p^{\alpha}\otimes p^{\alpha}$ system, there exists a local unitary operator $W$ which can realize the following UC-transformations:
\begin{equation*}
Z^{p^{s}}\stackrel{{W}}\sim Z^{p^{s}}\text{ and }X^{p^{t}}\stackrel{{W}}\sim X^{kp^{\alpha - s} + p^{t}},
\end{equation*}
where $s,t,k$ are non-negative integers and $s+t<\alpha,1\leq k<p^{s}$.
\footnote{In the rest of this paper, all of parameters are supposed to be non-negative integers if no special instructions.}
\label{conjugate transformation}
\end{lemma}

\begin{proof}
We firstly construct a specific operator $W$, then prove this operator realize the above transformations and $W$ is unitary.

Define $W$ as follows,
$$W|j+cp^{t}\rangle = |(j+c(kp^{\alpha - s}+p^{t}))\text{ mod }p^{\alpha} \rangle$$
 for each $j,c$, where $0\leq c <p^{\alpha - t}$ and $0\leq j< p^{t}$. Since it will be more brief for the proof if we have the matrix representation in the computational basis of $W$, suppose 
 $$W=\sum_{j,l\in\mathbb{Z}_{p^{\alpha}}}w_{j,l}|j\rangle \langle l|.$$
By the definition of $W$, we have $(\omega^{(j-l)p^{s}}-1) w_{j,l}=0$ for any $j,l$, which means that $WZ^{p^s}W^{\dagger}=Z^{p^{s}}$. Meanwhile, $w_{j,l}=w_{j+kp^{\alpha-s}+p^{t} ,l+p^{t}}$ for any $j, l$, thus $WX^{p^{t}}W^{\dagger}=X^{kp^{\alpha-s}+p^{t}}$.

Since there is exactly one 1 in each column of $W$, in the following we will prove $W$ is unitary by showing that there are no rows which has at least two non-zero elements. By contrary, suppose there exists a row which have at least two non-zero elements, \emph{i.e.}, there exist $i,j$ where $0\leq i\leq j<p^{t}$, and two different indices $i' = j_{1}p^{t} + i, j' =j_{2}p^{t} + j$, satisfying the following two equations,
\begin{align}
&w_{j_1(kp^{\alpha - s} + p^t) + i, i'}=w_{j_2(kp^{\alpha - s} + p^{t}) + j,j'}=1, \label{Eq-LUe-1}\\
&j_1(kp^{\alpha - s} + p^t) + i\equiv  j_2(kp^{\alpha - s} + p^{t}) + j \text{ (mod } p^{\alpha}) \label{Eq-LUe-2}
\end{align}
for some $0\leq j_1,j_2 <p^{\alpha - t}$.

Since Equation \eqref{Eq-LUe-2} is the same as $j-i\equiv (j_1-j_2)(kp^{\alpha - s - t} + 1)p^{t}\text{ (mod } p^{\alpha})$, which can not be satisfied when $s+t< \alpha$ and $i'\ne j'$. Thus the operator $W$ is unitary.
\end{proof}

Actually the operators which can realize such unitary conjugate transformation of Lemma \ref{conjugate transformation} might not be unique, and the operator $W$ we construct is a simple permutation operator. Next, we will discuss the properties of some group of unitary operators which realize similar UC-transformations with $W$, as a generalization of Clifford groups \cite{Appleby05} for the convenience of later classifications.

Consider GPM $M_{\mathbf{a}}$ in $p^{\alpha}\otimes p^{\alpha}$ system with form $$M_{\mathbf{a}}=X^{a_{1}p^{t}}Z^{a_{2}p^{s}},$$
where $ \mathbf{a}=(a_1,a_2)^T,$ and $1\leq s,t<\alpha$. Obviously, for any fixed $s,t$, the operators
$$\{\omega^{\theta p^{s+t}}M_{\mathbf{a}}\}$$
form a group, where $\theta$ is an arbitrary number, denoted by $\mathcal{G}(t,s)$. Let group $\mathcal{C}(t,s)$ be the normalizer of this group which contains all of the unitary operators that realize UC-transformations of $\mathcal{G}(t,s)$, \emph{i.e.}, $U\mathcal{G}(t,s)U^{\dagger}=\mathcal{G}(t,s)$ for $U\in \mathcal{C}(t,s)$. 

Similar to the properties of Clifford operators \cite{Appleby05}, we can write out the matrix representation for the UC-transformation made by the corresponding U in group $\mathcal{C}(t,s)$, and the determinant of this matrix equals to 1 module $p^{\alpha - s - t}$. Details are in Lemma \ref{LemGroup}. By \cite{JM14,Appleby05}, any Clifford operators can be uniquely represented by a $(2\times 2)$ symplectic matrix. When we restrict $U\in\mathcal{C}(t,s) $ $(s+t\leq \alpha)$, $U$ can be uniquely represented by a $(2\times 2)$ matrix $F$ with similar method.
%

\begin{lemma}\label{LemGroup}
Any operator $U\in\mathcal{C}(t,s) $ $(s+t\leq \alpha)$, can be presented by a unique ($2\times 2$) matrix $F$, and the determinant of $F$ satisfy
$$\det(F)\equiv 1\text{ (mod }p^{\alpha - s - t})$$
for which $$F=\begin{pmatrix}
\upsilon&\eta\\
\sigma&\tau
\end{pmatrix}$$
where $0\leq \upsilon,\eta< p^{\alpha - t}$, and $0\leq \sigma,\tau< p^{\alpha - s}$.
\end{lemma}

\begin{proof}
Since $U\in\mathcal{C}(t,s)$, then there exist $\upsilon,\eta,\sigma$ and $\tau$ such that
\begin{align}
X^{p^{t}}\mathop\sim\limits^{U}X^{\upsilon p^{t}}Z^{\sigma p^{s}},\text{ and }Z^{p^{s}}\mathop\sim\limits^{U}X^{\eta p^{t}}Z^{\tau p^{s}}
\label{Eq-UF}
\end{align}
where $0\leq \upsilon,\eta< p^{\alpha - t}$, and $0\leq \sigma,\tau< p^{\alpha - s}$. Thus we construct the presenting $(2\times 2)$ matrix $F$ of $U$ as
$$F\triangleq \begin{pmatrix}
\upsilon&\eta\\
\sigma&\tau
\end{pmatrix}$$
By Equation \eqref{Eq-UF}, $F$ is unique for which $0\leq \upsilon,\eta< p^{\alpha - t}$, and $0\leq \sigma,\tau< p^{\alpha - s}$.

Next, we prove $\det(F)\equiv 1\text{ (mod }p^{\alpha-s-t})$.  Let $\omega_1=\omega^{p^{s+t}}$, then there exists a function $g$ such that $$UM_{\mathbf{a}}U^{\dagger}=e^{ig(\mathbf{a})}M_{F\mathbf{a}}$$
Let $(F\mathbf{a})_{1}\triangleq \upsilon a_1+\eta a_2\text{ (mod } p^{\alpha -t})$, $(F\mathbf{a})_{2}\triangleq \sigma a_1+\tau a_2\text{ (mod } p^{\alpha -s})$, then we have
\begin{equation}
UM_{\mathbf{a}}M_{\mathbf{b}}U^{\dagger}=e^{i(g(\mathbf{a})+g(\mathbf{b}))}\omega_{1}^{(F\mathbf{a})_{2}\times ({F}\mathbf{b})_{1}}M_{{F}(\mathbf{a}+\mathbf{b})}
\label{a2}
\end{equation}
Since $\omega_{1}^{a_{1}b_{2}}M_{\mathbf{a}}M_{\mathbf{b}}=\omega_{1}^{a_{2}b_{1}}M_{\mathbf{b}}M_{\mathbf{a}}$, combining with Equation \eqref{a2}, we have $\omega_{1}^{a_{1}b_{2}-a_{2}b_{1}}=\omega_{1}^{(F\mathbf{a})_{1}\times (F\mathbf{b})_{2}-(F\mathbf{a})_{2}\times (F\mathbf{b})_{1}}$. Thus $\det(F) \equiv 1\text{ (mod }p^{\alpha - s - t})$, which completes our proof.

\end{proof}

It is easy to find that the unitary operator $W$ in Lemma \ref{permutation} is an element of $\mathcal{C}(t,s)$, and the corresponding $F$ of $W$ is
\begin{equation*}
F=\begin{pmatrix}
1&0\\
0&kp^{\alpha - s - t}+1
\end{pmatrix}
\end{equation*}
which satisfies $\det(F)\equiv 1 \text{ (mod } p^{\alpha - s- t})$.

Moreover, we can employ the contraposition of Lemma 2 to prove the UC-inequivalence of two GPM pairs. That is, given two UC-transformations (such as Equation \eqref{Eq-UF}) in which GPMs are in $\mathcal{G}(t,s)$, and the determinant of the corresponding $F$ is not equal to 1, then there are no unitary operators can transform these UC-transformations simultaneously.
 
\subsection{\label{sec:level3-2}Necessary conditions for U-equivalence of GPM-sets}

The above subsection provides some unitary operators as a transformation tool for U-equivalent GPM-sets. In this subsection, we will present two necessary conditions for U-equivalent GPM-sets, which can serve as a tool to prove the completeness of our classification.

Firstly, just as has been referred in \cite{Tian16}, we also need the U-equivalent invariants to be a necessary condition, that is, two U-equivalent GPM-sets must have the same value for those invariant. In a $d\otimes d$ quantum system, for GPM-set $\mathcal{M}=\{M_{1},\cdots,M_{n}\}$, let $\Delta_{ij}=M_{i}^{\dagger}M_{j}$ for $1\leq i,j\leq n$. Those three invariants in \cite{Tian16} are as follows.
\begin{eqnarray}
\label{invariant}
\begin{array}{c}
\vspace{0.2cm}
\mathbb{I}^{(1)}_{\mathcal{M}}=\frac{1}{d}\sum_{i,j,k,l\in[n]}\text{Tr}(|[\Delta_{ij},\Delta_{kl}]|^{2}) \\
\vspace{0.2cm}
\mathbb{I}^{(2)}_{\mathcal{M},a}=\frac{1}{d}\sum_{i,j\in[n]}|\text{Tr}((\Delta_{ij})^a)| \\
\vspace{0.2cm}
\mathbb{I}^{(3)}_{a,\mathcal{M}}=\frac{1}{d^{2}}\sum_{i,j,u,v,w,l\in[n]}|(\text{Tr}(\Delta_{ij}^{a}\Delta_{uv})\text{Tr}(\Delta_{ij}^{1-a}\Delta_{wl}))|
\end{array}
\end{eqnarray}
where $0<a<d$, $|\bullet|^2$ is the module square operator function, that is, $|A|^2 = AA^{\dagger}$. It is not hard to find the basic idea of these three invariants. Since $U_L\mathcal{N}U_R\approx \mathcal{M}$, without loss of generality, suppose $U_L N_i U_R \approx M_i$ for $1\leq i\leq n$, then $N_i^{\dagger}N_j\mathop\sim\limits^{U_{R}^{\dagger}} M_i^{\dagger}M_j$. Thus Tr$(N_i^{\dagger}N_j)=$ Tr$(M_i^{\dagger} M_j)$. Thus the above invariants hold. According to this idea, we can build two more invariants which is useful for our classification.

In a $d\otimes d$ quantum system, for any GPM-set $\mathcal{M}=\{M_{1},\cdots,M_{n}\}$, define $\mathcal{M}^t$ be the GPM-set $\{M^t_{1},\cdots,M^t_{n}\}$.
\begin{corollary}
The following two invariants
\begin{eqnarray}
\label{invariant2}
\mathbb{I}^{(1)}_{\mathcal{M}^{t}},\mathbb{I}^{(3)}_{a,\mathcal{M}^{t}},
\end{eqnarray}
hold when $\mathcal{M} \sim \mathcal{N}$ with $0<a,t<d$.
\end{corollary}

\begin{proof}
Since $U_{L}\mathcal{N}U_{R}\approx \mathcal{M}$, then $N_{i}^{\dagger}N_{j}\mathop\sim\limits^{U_{R}^{\dagger}} M_{i}^{\dagger}M_{j}$, thus we have
\begin{align}\label{EqInvPow}
\begin{aligned}
&U_{R}^{\dagger}(N_{i}^{t})^{\dagger}N_{j}^{t} U_{R}\\
\approx & U_{R}^{\dagger}N_{i}^{\dagger}N_{j}U_{R}U_{R}^{\dagger}N_{i}^{\dagger}N_{j}U_{R}\cdots U_{R}^{\dagger}N_{i}^{\dagger}N_{j}U_{R}\\
\approx & (M_{i}^{\dagger}M_{j})^{t}\\
\approx & (M_{i}^{t})^{\dagger} M_{j}^{t}
\end{aligned}
\end{align}
The first and third approximations of Equation \eqref{EqInvPow} hold since $\mathcal{M}$ and $\mathcal{N}$ are GPMs, they are commutative up to some phases. Then we have
\begin{eqnarray*}
\mathbb{I}^{(1)}_{\mathcal{M}^{t}} = \mathbb{I}^{(1)}_{\mathcal{N}^{t}},\mathbb{I}^{(3)}_{a,\mathcal{M}^{t}}=\mathbb{I}^{(3)}_{a,\mathcal{N}^{t}}.
\end{eqnarray*}
\end{proof}

The above invariants are quite helpful in the following classifications since the difference of any invariant leads to U-inequivalence. However, there still exist some exceptions when applying these invariants to GPM-sets. That is, two specific GPM-sets, which have the same value of all the above invariants, are proved to be U-inequivalent. At this time, we have to hunt for the second necessary conditions to explain the U-inequivalence between two GPM-sets.

Fortunately, for all the coming undermined U-equivalence in the classifying process, it is enough to prove the U-inequivalence of two specific sets just as shown in the following Theorem \ref{thm-inequivalence}. In the following theorem, we show two GPM triples are U-inequivalent while all the invariants in \eqref{invariant} and \eqref{invariant2} are equal in this case.

\begin{theorem}
In $p^{\alpha}\otimes p^{\alpha}$ quantum system, if the GPM-triples $\{I,Z^{p^s},X^{kp^t}Z^{t'p^{s}}\}$ are U-equivalent to $ \{I,Z^{p^s},X^{-kp^t}Z^{t'p^{s}}\}$, then one of the following items hold:
 \begin{itemize}
 \item[(a)]$ t'\in \{2,\frac{p^{t-s}+1}{2},p^{t-s}-1\}, p\geq 3$
 \item[(b)]$ t'\in\{2,p^{t-s}-1\}, p=2 \text{ and } s+t<\alpha -1  $
 \item[(c)] $p=2$ and $s+t+1=\alpha$.
\end{itemize}
where $s+t<\alpha,s<t,1< t'<p^{t-s}$, $k\perp p$.
\label{thm-inequivalence}
\end{theorem}

\begin{proof}
In order to show the U-equivalence of any two GPM-triples, we will at first simplify the U-equivalence of GPM-triples to UC-equivalence for two specific GPM-pairs, and then use the necessary condition in Lemma \ref{LemGroup} to deduce the UC-inequivalence for the two specific GPM-pairs.

Suppose GPM triples $\mathcal{N}=\{I,Z^{p^s},X^{kp^t}Z^{t'p^{s}}\}$ and $\mathcal{M}=\{I,Z^{p^s},X^{-kp^t}Z^{t'p^{s}}\}$ are U-equivalent. There exist unitary operators $U_1,U_2$, and $U_1\mathcal{N}U_2\approx\mathcal{M}$. There are 6 possible permutations for this transformation. We enumerate them as following (the left $i$-th GPM are unitary equivalent to the right $i$-th GPM for $1\leq i\leq 3$):
\begin{equation}
U_1\langle I,Z^{p^s},X^{kp^{t}}Z^{p^st'}\rangle U_2\approx\langle I,Z^{p^s},X^{-kp^{t}}Z^{p^st'}\rangle
\label{mode-1}
\end{equation}
\vspace{-1.3cm}

\begin{equation}
U_1\langle I,Z^{p^s},X^{kp^{t}}Z^{p^st'}\rangle U_2\approx\langle I,X^{-kp^{t}}Z^{p^st'},Z^{p^s}\rangle
\label{mode-2}
\end{equation}
\vspace{-1.3cm}

\begin{equation}
U_1\langle I,Z^{p^s},X^{kp^{t}}Z^{p^st'}\rangle U_2\approx\langle Z^{p^s},I,X^{-kp^{t}}Z^{p^st'}\rangle
\label{mode-3}
\end{equation}
\vspace{-1.3cm}

\begin{equation}
U_1\langle I,Z^{p^s},X^{kp^{t}}Z^{p^st'}\rangle U_2\approx\langle Z^{p^s},X^{-kp^{t}}Z^{p^st'},I\rangle
\label{mode-4}
\end{equation}
\vspace{-1.3cm}

\begin{equation}
U_1\langle I,Z^{p^s},X^{kp^{t}}Z^{p^st'}\rangle U_2\approx\langle X^{-kp^{t}}Z^{p^st'},Z^{p^s},I\rangle
\label{mode-5}
\end{equation}
\vspace{-1.3cm}

\begin{equation}
U_1\langle I,Z^{p^s},X^{kp^{t}}Z^{p^st'}\rangle U_2\approx\langle X^{-kp^{t}}Z^{p^st'},I,Z^{p^s}\rangle
\label{mode-6}
\end{equation}

For briefness, we only explain the details of the permutations \eqref{mode-1},\eqref{mode-3} as the examples, and the rest permutations can be analyzed in the same way.

Employing the three equivalence relationships in \eqref{mode-1}, we can easily eliminate $U_{2}$ and derive the following two UC-transformations under $U_{1}$:
 \begin{equation*}
\small
Z^{p^{s}}\mathop\sim\limits^{U_1} Z^{p^{s}}\text{ and }X^{kp^{t}}Z^{p^st'}\mathop\sim\limits^{U_1}X^{-kp^{t}}Z^{p^st'}.
\end{equation*}
Thus, we have $X^{p^{t}}\mathop\sim\limits^{U_1} X^{-p^{t}}$.
So $U_1$ is an element of $\mathcal{C}(t,s)$. The corresponding $(2\times 2)$ matirx $F$ of $U_{1}$ is as follows,
\begin{equation*}
F=\begin{pmatrix}
p^{\alpha - t}-1&0\\
0&1
\end{pmatrix}.
\end{equation*}
It is easy to show $M_{\mathbf{a}}\mathop\sim\limits^{U_1} M_{F\mathbf{a}}$, where $\mathbf{a} = (-1,1)^{T}$. By Lemma~\ref{LemGroup}, we have 
$$\det(F)\equiv 1\text{ (mod } p^{\alpha-s-t}),$$
which means $-1\equiv 1 \text{ (mod } p^{\alpha-s-t})$. It can be satisfied only if $p=2$ and $s+t+1=\alpha$.

In permutation mode \eqref{mode-3}, after eliminating $U_2$, we obtain the following two UC-transformations under $U_1$:
\begin{equation*}
Z^{p^{s}}\mathop\sim\limits^{U_1} Z^{-p^{s}}\text{ and }X^{kp^{t}}Z^{p^st'}\mathop\sim\limits^{U_1} X^{-kp^{t}}Z^{(t'-1)p^{s}}
\end{equation*}
From these two transformations, we have
\begin{equation}
X^{kp^{t}}\mathop\sim\limits^{U_1} X^{-kp^{t}}Z^{(2t'-1)p^{s}}.
\label{e2-1}
\end{equation}

If $2t'-1\not\equiv 0\text{ (mod }p^{t-s})$, let $k'p^{g}=2t'-1$ where $k'\perp p$ and $g<t-s$, let $u = t - s - g$, then
\begin{equation}\label{e2-Half}
X^{-kp^{t}}Z^{(2t'-1)p^{s}}\widesim[3]{V^{kk'^{-1}p^{u}}} Z^{(2t'-1)p^{s}}
\end{equation}
On the other hand, $X^{p^{t}}\widesim{R} Z^{p^{t}}$, combined with Equation \eqref{e2-1} and \eqref{e2-Half} we will get
\begin{equation}
\label{Eign}
Z^{p^{t}}\widesim{U}Z^{(2t'-1)p^{s}}
\end{equation}
for some unitary $U$. Meanwhile, $Z^{p^{t}}$ and $Z^{p^{s}}$ are not UC-equivalent for $s\ne t$, then $(2t'-1)\equiv 0\text{ (mod }p^{t-s})$ if transformation \eqref{e2-1} holds. Thus the necessary condition of this permutation mode is $t'=(p^{t-s}+1)/2$.

For the remaining four cases, we find the determinant $\det(F)\text{ module }p^{\alpha -s-t}$ for the corresponding $(2\times 2)$ matrices $F$ of $U_{1}$ in permutations \eqref{mode-2}, \eqref{mode-4}, \eqref{mode-5}, \eqref{mode-6} are equal to $1,-1,1,-1$ when module $p^{\alpha -s -t}$ respectively. By Lemma \ref{LemGroup}, permutations \eqref{mode-4}, \eqref{mode-6} can be satisfied only when $p=2$ and $s+t+1=\alpha$. Combining with the results of permutations \eqref{mode-2} and \eqref{mode-5} (we put them into Appendix \ref{appendix3}), we find that the above GPM-triples are U-equivalent if $t' \in\{2,\frac{p^{t-s}+1}{2},p^{t-s}-1\}$ when $p\geq 3$ or $t'\in\{2,p^{t-s}-1\} $ when $p=2,s+t<\alpha-1$ or $p=2,s+t+1=\alpha$.
\end{proof}

Now, on the basis of useful operators for U-transformations and necessary conditions for U-equivalent GPMs, we are ready to discuss the U-classifications of GPM-sets in the corresponding quantum systems.

\section{\label{sec:level4} Classification of GPM-pairs}
In this section, we will classify all the GPM-pairs in $d\otimes d$ quantum system into U-inequivalent classes completely. GPM-pairs are the simplest GPM-sets, and studying their U-classifications will be the first step of U-equivalence of GPM-sets without doubt.

Based on the conditions for U-equivalence of sets just mentioned in the above section, we can classify all the GPM-pairs $\{I,X^{s}Z^{t}\}$ (the first one can always be U-equivalent to identity while the second one keeps GPM form) into minimal U-equivalence classes for $s,t<d$. Specifically, we have the following theorem.

\begin{theorem}
All of the GPM-pairs can be classified into $\prod_{1\leq i\leq n}(k_{i}+1)$ U-inequivalent pairs in $d\otimes d$ quantum system, where $d=p_{1}^{k_{1}}p_{2}^{k_{2}}\cdots p_{n}^{k_{n}}$ for some different primes $p_1,\cdots,p_n$. Furthermore, the representing GPM pairs are $\{I,Z^{s}\}$, where $s$ is a factor of $d$.
\label{thm-1}
\end{theorem}
\begin{proof}
Actually, we only need to consider pairs $\{I,X^sZ^t\}$. Since for any $0\leq s,t<d$,
$$\begin{aligned}
\{I,X^{s}Z^{t}\}\mathop\sim\limits^{\text{\cite{Nielsen02}}} \{I,Z^{\text{gcd}(s,t)}\}
\end{aligned}$$
\emph{i.e.}, there exists a unitary operator $U$ such that $UX^sZ^tU^{\dagger}\approx Z^{\text{gcd}(s,t)}$. Meanwhile,
$$\mathcal{M}_{s}\triangleq \{I,Z^{s}\}\mathop\sim\limits^{Q_{k}}\{I,Z^{sk}\}$$
 for any $k\perp d$.

What's more, $\{I,Z^{s}\}\not\sim \{I,Z^{t}\}$ for the reason that
$$\mathbb{I}^{(2)}_{\mathcal{M}_s,\frac{d}{s}}=4,\quad \mathbb{I}^{(2)}_{\mathcal{M}_t,\frac{d}{s}}=2$$
where $s,t$ are two different factors of $d$.

Therefore, we get minimum U-equivalence classes $\mathcal{M}_{s}$ in which $s$ is invertible and $1\leq s<d$ in any $d\otimes d$ quantum system. Thus there are totally $\prod_{1\leq i\leq n}(k_{i}+1)$ different U-inequivalent GPM-pairs in $d\otimes d$ quantum system.
\end{proof}

It's amazing that the U-inequivalent GPM-pairs in $d\otimes d$ quantum system only depends on the number of factors of $d$, rather than the factors of itself. 

 \section{\label{sec:level5} Classification of GPM-triples in $p^{\alpha}\otimes p^{\alpha}$ quantum system}

Tian \emph{et al.} gave a classification of GPM triples in $d\otimes d$ quantum system when $d$ is a prime or $d\in\{4,6\}$~\cite{Tian16}. The biggest challenge to generalize the dimension from a prime to a composite is the degenerate of the operators. We use some ingenious methods including splitting the degenerate GPM operators into the form of tensor of two lower dimensional operators to find the U-equivalent relationships between two GPM-triples. Meanwhile, we analyze the special properties of some specific degenerate GPM-triples by restricting the basic operator from $X_{p^{\alpha}}, Z_{p^{\alpha}}$ into $X_{p^{\alpha}}^{p^{s}}, Z_{p^{\alpha}}^{p^{t}}$ where $0\leq s,t < \alpha$, such as Theorem \ref{thm-inequivalence}.

In this section, we will give a classification of GPM-triples in any $p^{\alpha}\otimes p^{\alpha}$ quantum system, where $p$ is a prime and $\alpha\geq 2$. Before giving a classifiction of GPM-triples in $p^{\alpha}\otimes p^{\alpha}$ quantum system, we firstly introduce two equivalence classes $[x]_{d},[[ x ]]_{d}$ where $x\in\mathbb{Z}_{d}$ and $x\geq 2$, which are defined to characterize the U-equivalence of GPM-triples more effectively, in which
\begin{align*}
\begin{aligned}
 &[x]_{d}=\{x,x^{-1},1-x,(1-x)^{-1},1-x^{-1},x(x-1)^{-1}\},\\
 &[[ x ]]_{d}=\{x,1-x^{-1},(1-x)^{-1}\}.
 \end{aligned}
 \end{align*}
where $x^{-1}$ is the inverse of $x$ in $\mathbb{Z}_{d}$. We also omit the subscript $d$ of $[x]_{d},[[x]]_{d}$ when there is no ambiguity or $d$ is exactly the dimension of quantum system. The first equivalence class is introduced by \cite{Tian16}, which is served as a partition of some kind of GPM-triples, and we introduce an extra equivalence class $[[ x ]]$ to give a more careful partition. If $x^{-1}$ or $(1-x)^{-1}$ does not exist, we just omit this element. For example, when $d=9, [3]=\{3,4,6,7\}$ and $[[3]]=\{3,7\}$.


It is natural to consider the number of elements in each $[x]$ or $[[x]]$ for counting U-inequivalence classes in $p^{\alpha}\otimes p^{\alpha}$ quantum system. When $x\not\in[2]$, all of elements in $[x]$ (or $[[x]]$) are distinct if $x \ne x^{-1}$ and $x \ne 1 - x^{-1}$. Fortunately, the case $x=x^{{-1}}$ or $x=1-x^{-1}$ comes up not so frequently: $x = x^{-1}$ only when $p = 2$ and $x\in\{2^{\alpha-1}+1,2^{\alpha-1}-1\}$ ($\alpha\geq 2$); 
\begin{equation}\label{eq-Inv}
x = 1 - x^{-1}
\end{equation}
iff $p$ has form $6k+1$. Details are as shown in Lemma \ref{LemInv}.



\begin{lemma}\label{LemInv}
For $x\in \mathbb{Z}_{p^{\alpha}}$ and $p^{\alpha}\ne 3$, there exactly exist two solutions of Equation \eqref{eq-Inv} if $p = 6k+1$ for some $k$, and there are no solutions of Equation \eqref{eq-Inv} if $p \ne 6k+1$.
\end{lemma}

\begin{proof}
We will prove a stronger conclusion: If $p = 6k+1$ for some $k$, there exactly exist two solutions of Equation \eqref{eq-Inv} and the two solutions are not in the set $\{x|x\equiv 2\text{ (mod }p)\}$, otherwise there are no solutions of Equation \eqref{eq-Inv}. 

Firstly, there are no solutions of Equation \eqref{eq-Inv} if $p = 2$, since $ x\equiv 1\text{ (mod } 2)$ when $p = 2$ and $x$ is invertible. 

Suppose $x$ is the solution of Equation \eqref{eq-Inv}. Since $p^{\alpha}\ne 3$, then $x\ne x^{-1}$ and there are even solutions of Equation \eqref{eq-Inv}.

When $p>2$ and $p^{\alpha}\ne 3$, we will prove this conclusion by induction on $\alpha$. Firstly, when $\alpha=1$, we have $p\geq 5$. Thus $p$ has form $6k+1$ or $6k+5$ for some $k$. Observe that Equation 
\begin{equation}
(2x-1)^{2}\equiv -3\text{ (mod }{p^{\alpha}})
\label{Eq-Quadratic}
\end{equation}
is equivalent to Equation \eqref{eq-Inv}, thus we will solve Equation \eqref{Eq-Quadratic} instead of Equation \eqref{eq-Inv}. By quadratic residue theorem we know $-3$ is a quadratic residue of $p$ iff $p=6k+1$ for some $k$. Thus when $\alpha$ equals 1, there are exactly two solutions for Equation \eqref{Eq-Quadratic} when $p=6k+1$, and no solutions for Equation \eqref{Eq-Quadratic} when $p=6k+5$. It is easy to check that $2$ is not the solution. Meanwhile, when $p = 3$ and $\alpha = 2$ there are no solutions of Equation \eqref{eq-Inv}.

Suppose the conclusion holds for some $\alpha$ where $\alpha\geq 1$ or $p = 3, \alpha\geq 2$, and we will prove the conclusion also holds for $\alpha+1$.

For $x\in\mathbb{Z}_{p^{\alpha+1}}$, suppose $y$ is the solution of Equation \eqref{eq-Inv}, then $y'\triangleq y^{-1}$ is also the solution of Equation \eqref{eq-Inv}. Let 
$$s_{1}\triangleq y \text{ (mod }p^{\alpha}),s_{2}\triangleq y'\text{ (mod }p^{\alpha}).$$
It is easy to check $s_{1},s_{2}$ are two different solutions of Equation \eqref{eq-Inv} for $x\in\mathbb{Z}_{p^{\alpha}}$. 
Thus $y,y'$ have forms
\begin{align}\label{soluAlp1}
\begin{aligned}
& y = c_{1}p^{\alpha}+s_1\text{ (mod }p^{\alpha+1})\\
& y' = c_{2}p^{\alpha}+s_2\text{ (mod }p^{\alpha+1})
\end{aligned},
\end{align}
for some $c_{1},c_{2}$. Since $s_{1}, s_{2}\not\in \{x|x\equiv 2 \text{ (mod }p)\}$, and $s_1-s_2\equiv 2s_1-1\text{ (mod }p)$, then $s_{1}, s_{2}\not\in \{x|x\equiv 2^{-1} \text{ (mod }p)\}$ and $s_{1}-s_{2}$ is invertible in $\mathbb{Z}_{p}$. By solving Equation \eqref{eq-Inv} and \eqref{soluAlp1}, there only exists one pair $y,y'$ in which $y,y'<p^{\alpha+1}$, and the value of $y,y'$ are as follows,
$$\begin{aligned}
&y = (\sigma-s_1)(s_1-s_2)^{-1}p^{\alpha}+s_1\text{ (mod }p^{\alpha+1}),\\
& y' = (-1-(\sigma-s_1)(s_1-s_2)^{-1})p^{\alpha}+s_2\text{ (mod }p^{\alpha+1}),
\end{aligned}$$
where the inverse of $(s_{1}-s_{2})$ is operating in the field $\mathbb{Z}_{p}$.
Thus there are exactly two solutions in $\mathbb{Z}_{p^{\alpha + 1}}$ iff $p = 6k+1$. Since $y\equiv s_1\text{ (mod }p)$, then $y,y'\not\in\{x|x\equiv 2\text{ (mod }p)\}$. Thus we are done!
\end{proof}

Next we will firstly give a classification of GPM-triples in $p^2\otimes p^2$ quantum system for simplification. The U-transformation operators are mainly composed of Clifford operators and some simple GPMs. Nevertheless, there exist some GPM-triples for which all of the above operators can not serve as its unitary transformation tools, fortunately, we can deal with these GPM-triples by splitting the GPM as a tensor of two separate parts. Afterwards, we prove the completeness of our classification mainly by invariants \eqref{invariant} and Theorem \ref{thm-inequivalence}.
\begin{theorem}
There are exactly $\lfloor\frac{5}{6}p^{2}\rfloor + \lfloor \frac{p-2}{6}+(-1)^{\lfloor \frac{p}{3}\rfloor}\frac{p}{3}\rfloor + 3$ U-inequivalent GPM-triples in $p^{2}\otimes p^{2}$ quantum system for prime $p>2$. Furthermore, the representing GPM-triples are as follows,
\begin{itemize}
\item[\textcircled{1}] $\mathcal{M}_{s_{1}}=\{I,Z,X^{s_{1}}\}$, where $1\leq s_{1} \leq \lfloor\frac{p^{2}}{2}\rfloor$.
\item[\textcircled{2}] $\mathcal{N}_{k,\hat{s}_{2}}=\{I,Z,X^{kp}Z^{\hat{s}_{2}}\}$, where $1\leq k\leq \frac{p-1}{2}, \hat{s}_{2}$ is the minimal element of $[[s_{2}]]_{p},$ and $ 2\leq s_{2}\leq p-1$. There are totally $\lfloor\frac{p-1}{2}\rfloor(2\lfloor\frac{p}{6}\rfloor+1)$ U-inequivalent classes in such forms.
\item[\textcircled{3}] $\mathcal{N}_{\hat{s}_{3}}=\{I,Z,Z^{\hat{s}_{3}}\}$, where $\hat{s}_{3}$ is the minimal element of $[s_{3}]$, and $2\leq s_{3}<p^{2}$. There are totally $\lfloor\frac{p^{2}-2p+1}{6}\rfloor+\lfloor\frac{p-1}{2}\rfloor+1$ U-inequivalent classes in such forms.
\item[\textcircled{4}] $\mathcal{M}=\{I,Z^{p},X^{p}\}$.
\item[\textcircled{5}] $\mathcal{W}_{\hat{s}_{4}}=\{I,Z^{p},Z^{\hat{s}_{4}p}\}$, where $\hat{s}_{4}$ is the minimal element of $[s_{4}]_{p}$ and $2\leq s_{4}\leq p-1$. There are totally $\lfloor\frac{p}{6}\rfloor+1$ LU-inequivalent classes in such forms.
\end{itemize}
\label{thm3}
\end{theorem}

\begin{proof}
By the similar classification of Theorem \ref{thm-1}, all of GPM-triples can be divided into two cases: $\{I,Z,X^{s}Z^{t}\},\{I,Z^{p},X^{s}Z^{t}\}$. We find the representing GPM-triples of Theorem \ref{thm3} by Table \ref{EqGPM}. For example, for GPM-sets $\{I,Z,X^{s}Z^{t}\}$ where $s\perp p$, we can use U-equivalence transformation \eqref{eqX} followed by \eqref{eqIX} to transform it to be the class $\mathcal{M}_{s}(1\leq s\leq \frac{p^{2}}{2})$. This belongs to case \textcircled{1} in Theorem \ref{thm3}.

To prove the classification of the theorem is minimum, we need to show:

(1) The GPM-triples are U-inequivalent when the parameters are distinct for the internal classes of each cases.

(2) Every two cases are U-inequivalent.

Here, we firstly prove the U-inequivalence for internal classes of case \textcircled{2}, and leave the other cases in Appendix \ref{appendix1}. Then, we show the U-inequivalence between any two cases.

\textbf{Case \textcircled{2}} Since
\begin{equation}
\small
\mathbb{I}_{\mathcal{N}_{k,s}}^{(1)}=48(1-\cos (2kp\pi/p^{2}))
\label{eq3.5}
\end{equation}
we have 
$$\mathbb{I}_{\mathcal{N}_{k,s}}^{(1)}\ne \mathbb{I}_{\mathcal{N}_{k',s'}}^{(1)},$$ 
where $1\leq k<k'\leq \lfloor\frac{p}{2}\rfloor$. Notice the value of $\mathbb{I}_{\mathcal{N}_{k,s}}^{(1)}$ in Equation \eqref{eq3.5} depends only on $k$, thus $\mathcal{N}_{k,s}\not\sim \mathcal{N}_{k',s'}$, where $2\leq k<k'\leq \frac{p-1}{2}$.

Since
$$\mathbb{I}_{\mathcal{N}_{k,s}^{p},s}^{(3)}\ne \mathbb{I}_{\mathcal{N}_{k,s'}^{p},s}^{(3)}$$
 for any $s'\not\in[s]_{p}$, we have $\mathcal{N}_{k,s}\not\sim \mathcal{N}_{k,s'}$ where $s'\not\in[s]_{p}$. Meanwhile, $[2]_{p}=[[ 2]]_{p}=\{2,\frac{p+1}{2},p-1\}$, thus we need only prove $\mathcal{N}_{k,s}$ is U-inequivalent to $\mathcal{N}_{k,\text{Inv}_p(s)}$ for $s>2$ and $s\perp p$. In this case, $\mathcal{N}_{k,\text{Inv}_p(s)}$ is U-equivalent to $\mathcal{N}_{-k,s}$ by U-transformations \eqref{eqXpZ} and \eqref{eqXpZ2}. By Theorem \ref{thm-inequivalence} (Set $\alpha=2,s=0,t=1$ for the theorem \ref{thm-inequivalence}), we find $\mathcal{N}_{k,t}\sim\mathcal{N}_{-k,t}$ only if $t\in[[2]]_{p}$ since $p>2$.

Next, we need to count the number of U-inequivalence classes with form \textcircled{2}. Suppose there are $\Delta$ U-inequivalence classes for certain $k$. It's easy to find $\Delta=0$ when $p$ equals $2$, and there is only one GPM-triple $\{I,Z,X^{3}Z^{2}\}$ when $p$ equals 3.

For $p\geq 5$, if $p=6k+5$, each equivalence class $[[ s]]_{p}$ has three elements by Lemma \ref{LemInv}, \emph{i.e.}, $p=2+3\Delta$, thus $\Delta=2k+1$. If $p=6k+1$, there are two classes $[[ s]]_p$ and $[[ s^{-1}]]_{p}$ which has only one element respectively (By Lemma \ref{LemInv}), \emph{i.e.}, the three elements in $[[s]]_{p}$ are equal to each other. Then we have $p=2+3(\Delta-2)+2$, thus $\Delta=2k+1$ in this case.

\begin{widetext}
\begin{flalign}
&\{I,Z^{a},X^{s}Z^{t}\}\mathop\sim\limits^{P^{-\mu}}\{I,Z^{a},X^{s}\}, \text{ where }\mu=t\cdot s^{-1} \text{ if }s\perp p, \text{ and } \mu=\frac{t}{p}\cdot (\frac{s}{p})^{-1} \text{ if } p|t \text{ and } p|s. \label{eqX}\\
&\{I,Z,X^{s}\}\mathop\sim\limits^{R}\{I,Z^{s},X^{-1}\}\widesim{P^{-s}}\{I,Z^{s},X^{-1}Z^{s}\}\sim\{I,X^{-1},Z^{-s}\}\widesim{R^{-1}}\{I,Z,X^{-s}\} \text{ for any } s.\label{eqIX}\\
&\{I,Z,X^{kp}Z\}Z^{-1}\approx\{Z^{-1},I,X^{kp}\}\mathop\sim\limits^{Q_{-1}}\{I,Z,X^{-kp}\}\text{ for any }k.\label{eqXp1}\\
&\{I,Z,X^{kp}Z^{t}\}\widesim{P^{-k^{-1}}}\{I,Z,X^{kp}Z^{t-p}\},\text{ where }k\perp p.\label{eqXpZ2}\\
&\begin{cases}
\{I,Z,X^{kp}Z^{s}\}\mathop\sim\limits^{Q_{-1}} \{I,Z^{-1},X^{-kp}Z^{-s}\}\sim \{I,Z,X^{-kp}Z^{1-s}\}\text{ for any } k,s.\\
\{I,Z,X^{kp}Z^{s}\}\widesim{Q_{s^{-1}}}\{I,Z^{s^{-1}},X^{skp}Z\}\widesim{V^{-skp}}\{I,Z,X^{-kp}Z^{s^{-1}}\},\text{ where }s\perp p.
\label{eqXpZ}
\end{cases}\\
&\{I,Z^{p},X^{s}\}\mathop\sim\limits^{R}\{I,Z^{s},X^{-p}\}\widesim{Q_{s^{-1}}}\{I,Z,X^{-sp}\}, \text{ where }s\perp p.\label{eqZp1}\\
&\begin{cases}
\{I,Z^{a},Z^{ka}\}\widesim{Q_{k^{-1}}} \{I,Z^{a},Z^{ak^{-1}}\},\text{ where } k\perp p.\\
\{I,Z^{a},Z^{ka}\}\mathop\sim\limits^{Q_{-1}}\{I,Z^{-a},Z^{-ka}\}\sim\{I,Z^{a},Z^{(1-k)a}\}
\label{eqZ}
\end{cases}\\
&\{I,Z^{p},X^{kp}Z^{t}\}\widesim{V^{-kpt^{-1}}}\{I,Z^{p},Z^{t}\}\widesim{Q_{t^{-1}}} \{I,Z,Z^{pt^{-1}}\}=\{I,Z,Z^{pt^{-1}}\}\label{eqXp2},\text{ where }t\perp p.\\
&\{I,Z_{p^2}^{p},X_{p^2}^{kp}\}=\{I_{p}\otimes I_{p},I_{p}\otimes Z_{p},X_{p}^{k}\otimes I_{p}\}\widesim{Q_{k}\otimes I_{p}}\{I_{p}\otimes I_{p},I_{p}\otimes Z_{p},X_{p}\otimes I_{p}\}=\{I,Z_{p^2}^{p},X_{p^2}^{p}\}\text{ for any }k.\label{eqXp}
\end{flalign}
\end{widetext}

\begin{widetext}
\begin{center}
\begin{table}[H]
\caption{\mbox{Classification of GPM triples in $p^2\otimes p^{2}$ quantum system.}}
\begin{tabular}{ c | c | c | c |c |c}
	\hline\hline
    GPM-sets & \multicolumn{2}{c|}{some cases of $s$ and $t$} & U-equivalences & U-equivalence classes & In Th. \ref{thm3}\\
    \hline
    \multirow{4}{*}{$\{I,Z,X^{s}Z^{t}\}$} &\multicolumn{2}{c|}{$s\perp p$ or $(p|s,p|t)$ or $t=0 $} & \eqref{eqX},\eqref{eqIX} & $\mathcal{M}_{s}(1\leq s\leq \frac{p^{2}}{2})$ & \textcircled{1}\\
    \cline{2-6}
    & \multirow{2}{*}{$s=kp(k\perp p),t\perp p$} & $t=1$ & \eqref{eqXp1},\eqref{eqIX} & $\mathcal{M}_{kp}(1\leq kp\leq p^{2}/2) $ &\textcircled{1}\\
    \cline{3-6}
    & & $t\geq 2$ & \eqref{eqXpZ},\eqref{eqXpZ2} & $\mathcal{N}_{k,\hat{t}}$($\hat{t}$: Min of $[[ t]]_{p}$) & \textcircled{2}\\
    \cline{2-6}
    &\multicolumn{2}{c|}{$s=0$} & \eqref{eqZ} & $\mathcal{N}_{\hat{t}}$($\hat{t}$: Min of $[t^{-1}]$) & \textcircled{3}\\
    \hline
    \multirow{5}{*}{$\{I,Z^{p},X^{s}Z^{t}\}$} & \multicolumn{2}{c|}{$s\perp p$} & \eqref{eqX},\eqref{eqZp1},\eqref{eqIX} & $\mathcal{M}_{sp}(1\leq s\leq \frac{p^{2}}{2})$ & \textcircled{1}\\
    \cline{2-6}
    & \multirow{2}{*}{$s=0$} & $t\perp p$ & \eqref{eqXp2},\eqref{eqZ} & $\mathcal{N}_{a}$($a$: Min of $[pt^{-1}]$) & \textcircled{3}\\
    \cline{3-6}
    && $t=kp(k\perp p)$ & \eqref{eqZ} & $\mathcal{W}_{\hat{k}}$ ($\hat{k}$: Min of $[ k]_{p}$)& \textcircled{5}\\
    \cline{2-6}
    & \multirow{2}{*}{$s=kp$} & $t\perp p$ & \eqref{eqXp2},\eqref{eqZ} &$\mathcal{N}_{b}$($b$: Min of $[pt^{-1}]$) & \textcircled{3}\\
    \cline{3-6}
    && $p|t $ or $t=0$ & \eqref{eqX},\eqref{eqXp}&$\mathcal{M}$ &\textcircled{4}\\
\hline\hline
\end{tabular}
\label{EqGPM}
\end{table}
\end{center}\end{widetext}

Thus there are exactly $\frac{p-1}{2}\cdot(2\lfloor\frac{p}{6}\rfloor+1)$ U-inequivalence classes for any $p\geq 5$. When $p=3$ this result also holds.

Since there exist at least one of invariants which are distinct for any two cases of Theorem \ref{thm3} by Table \ref{GPMInequivalence}, then all of cases in Theorem \ref{thm3} are U-inequivalent. Thus the classification of Theorem \ref{thm3} are complete.
\begin{table}
\centering
\caption{Any two cases in Theorem \ref{thm3} are U-inequivalent (When $p=3$ and $\hat{s}_{4}=2, $ we have $ A=33$, otherwise $A=27$.)}
\begin{tabular}{ m{2.5cm}<{\centering} | m{2.5cm}<{\centering} | m{2cm}<{\centering}  |m{1cm}<{\centering} }
	\hline\hline
	Triples $\mathcal{T}$ & $\mathbb{I}_{\mathcal{T}}^{(1)}$ & $\mathbb{I}_{\mathcal{T},p}^{(2)}$ & $\mathbb{I}_{\hat{s}_{4},\mathcal{T}}^{(3)}$\\
	\hline
	$\mathcal{M}_{s_{1}}$ & $48(1-\cos\frac{2\pi s_1}{p^2})$ & 5 (for $p|s_{1})$ & -- \\
	\hline
	$\mathcal{N}_{k,\hat{s}_{2}}$ ($\hat{s}_{2}$: Min of $[[ s_{2}]]_{p}$) & $48(1-\cos\frac{2\pi k}{p})$ & 3 & -- \\
	\hline
    $\mathcal{N}_{\hat{s}_{3}}$ ($\hat{s}_{3}$: Min of $[s_{3}]$) & 0 & {$\leq 5$} & $\geq 27$  \\ \hline
    $\mathcal{M}$ & 0 & 9 & $A$ \\ \hline
    $\mathcal{W}_{\hat{s}_{4}}$($\hat{s}_{4}$: Min of $[s_{4}]_{p}$) & 0 & 9 & $>A$
     \\ \hline\hline
\end{tabular}
\label{GPMInequivalence}
\end{table}

\end{proof}

\begin{figure}[!ht]
\centering
\includegraphics[width=0.5\textwidth]{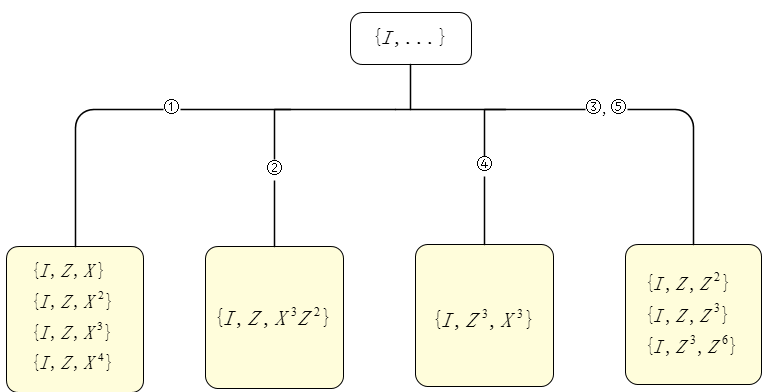}
\caption{Classification of GPM-triples in $9\otimes 9$ quantum system}\label{FigCD9}
\end{figure}

Our classification also holds for $4\otimes 4$ quantum system when we remove the trivial GPM-triples. To check the correctness of our classification, we enumerate the classification results of GPM-triples in $9\times 9$ quantum system in Figure \ref{FigCD9}. There are exactly 9 U-inequivalent GPM-triples in $9\otimes 9$ quantum system via Theorem \ref{thm3}, and all of triples in Figure \ref{FigCD9} are U-inequivalent via Table \ref{tableCD9}.
\begin{table}
\begin{center}
\caption{Invariants of the minimal triples in $9\otimes 9$ system. (All of the decimals are the approximate value.)}
\begin{tabular}{ c | c | c | c |c }
	\hline\hline
	Triples & $\mathbb{I}^{(1)}$ & $\mathbb{I}_{\mathcal{M},3}^{(2)}$ & $\mathbb{I}_{2,\mathcal{M}}^{(3)}$ & $\mathbb{I}_{2,\mathcal{M}^{3}}^{(3)}$ \\ \hline
	$\{I,Z,X\}$ & 11.23 & 3 & 27  &33 \\ \hline
	$\{I,Z,X^{2}\}$ & 39.67 & 3 & 27 &33 \\ \hline
    $\{I,Z,X^{3}\}$ & 72 & 5 & 29  & 141\\ \hline
    $\{I,Z,X^{4}\}$ & 93.11 & 3 & 27 &33 \\ \hline
    $\{I,Z,X^{3}Z^{2}\}$ & 72 & 3 & 27& 81 \\ \hline
    $\{I,Z^{3},X^{3}\}$ & 0 & 9 & 33  &729\\ \hline
    $\{I,Z,Z^{2}\}$ & 0 & 3 & 35 &81 \\ \hline
    $\{I,Z,Z^{3}\}$ & 0 & 5 & 31  &141\\ \hline
    $\{I,Z^{3},Z^{6}\}$ & 0 & 9 &  81 &729 \\ \hline\hline
\end{tabular}
\label{tableCD9}
\end{center}
\end{table}


Now we consider some general situations by generalizing the quantum system to $p^{\alpha}\otimes p^{\alpha}$, where $p$ is a prime.

\begin{theorem}
\begin{itemize}The GPM-triples in $p^{\alpha}\otimes p^{\alpha}$ system can be classified into the following classes, where $p$ is a prime and $p>2$.

\vspace{-0.2cm}
\item[\textcircled{1}]$\mathcal{M}_{(s_{1},t_{1})}=\{I,Z^{p^{s_{1}}},X^{p^{t_{1}}}\}$, where $1\leq s_{1}\leq t_{1}<\alpha$ and $s_{1}+t_{1}\geq \alpha$.
\vspace{-0.2cm}
\item[\textcircled{2}]$\mathcal{M}_{(k_{1},s_{2},t_{2})}=\{I,Z^{p^{s_{2}}},X^{k_{1}p^{t_{2}}}\}$, where $0\leq s_{2}\leq t_{2}<\alpha, s_{2}+t_{2}<\alpha, 1\leq k_{1}\leq \frac{p^{\alpha - s_{2} -t_{2}}}{2}$ and $k_{1}\perp p$.
\vspace{-0.2cm}
\item[\textcircled{3}]$\mathcal{N}_{(s_{3},\hat{t}_{3})}=\{I,Z^{p^{s_{3}}},Z^{p^{s_{3}}\hat{t}_{3}}\}$, where $\hat{t}_{3}$ is the minimal element of $[t_{3}]_{p^{\alpha-s_{3}}}, 0\leq s_{3}<\alpha,2\leq t_{3}<p^{\alpha - s_{3}}$.
\vspace{-0.2cm}
\item[\textcircled{4}]$\mathcal{N}_{(k_{2},s_{4},t_4,\hat{t}_{4}')}=\{I,Z^{p^{s_{4}}},X^{k_{2}p^{t_4}}Z^{p^{s_{4}}\hat{t}_{4}'}\}$, where $\hat{t}_{4}'$ is the minimal element of $[[ t_4']]_{t_4-s_{4}}, 0\leq s_{4}<t_4<\alpha, s_{4}+t_4<\alpha,1\leq k_{2}\leq \frac{p^{\alpha -s_{4}-t_4}}{2}, 2\leq t_4'<p^{t_4-s_{4}}$ and $k_{2}\perp p$.
\vspace{-0.2cm}
\item[\textcircled{5}]$\mathcal{N}_{(s_{5},t_5,\hat{t}_{5}')}=\{I,Z^{p^{s_{5}}},X^{p^{t_5}}Z^{p^{s_{5}}\hat{t}_{5}'}\}$, where $\hat{t}_{5}'$ is the minimal element of $[t_{5}']_{p^{t_{5}-s_{5}}}, 0\leq s_{5}<t_5<\alpha,s_{5}+t_5\geq \alpha, 2\leq t_5'<p^{t_5 -s_{5}}$.
\end{itemize}
\label{thm4}
When $p$ equals 2, the classifications are the same as the above results except case \textcircled{4}, which are distinct when $s+t+1=\alpha$. We list the new U-inequivalent classes of case \textcircled{4} for $2^{\alpha}\otimes 2^{\alpha}$ system in the following.

 \textcircled{4}.1 $\mathcal{N}_{(k_{2},s_{4},t_4,\hat{t}_{4}')}=\{I,Z^{2^{s_{4}}},X^{k_{2}2^{t_4}}Z^{2^{s_{4}} \hat{t}_{4}'}\}$, where $\hat{t}_{4}'$ is the minimal element of $[[ t_{4}']]_{2^{t_{4}-s_{4}}}, 0\leq s_{4}<t_4<\alpha, s_{4}+t_4+1<\alpha,1\leq k_{2}< 2^{\alpha -s_{4}-t_4-1}, 2\leq t_4'<2^{t_4-s_{4}}$ and $k_{2}\perp p$.

 \textcircled{4}.2 $\mathcal{N}_{(k,s,t,\hat{t}')}'=\{I,Z^{2^s},X^{k2^{t}}Z^{2^{s}\hat{t}'}\}$, where $\hat{t}'$ is the minimal element of $[t']_{2^{t-s}},0\leq s<t<\alpha, s+t+1=\alpha,1\leq k\leq 2^{\alpha -s-t-1},2\leq t'<2^{t-s}$ and $k\perp p$.

 Meanwhile, for any $p^{\alpha} \otimes p^{\alpha}$ quantum system, we find that the number of U-inequivalent triples are about $\frac{(\alpha + 3)}{6}p^{\alpha} + O(\alpha p^{\alpha-1})$.
\end{theorem}

\begin{figure}
\center
\includegraphics[width=0.5\textwidth]{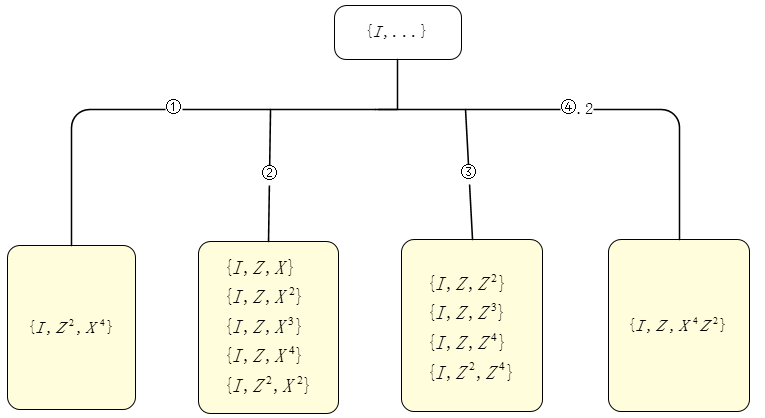}
\caption{Classification of GPM-triples in $8\otimes 8$ quantum system.}
\label{FigCD8}
\end{figure}
We list the classification of GPM-triples in $8\otimes 8$ quantum system to give an intuitive comprehension in figure \ref{FigCD8}. There are exactly 11 U-inequivalent GPM-triples, and all of triples in Figure \ref{FigCD8} are U-inequivalent via Table \ref{tableCD8}.

\begin{table}
\begin{center}
\caption{Invariants of the minimal triples in $8\otimes 8$ system. (All of the decimals are the approximate value.)}
\begin{tabular}{ c | c | c | c |c}
	\hline\hline
	Triples & $\mathbb{I}^{(1)}$  & $\mathbb{I}_{\mathcal{M},4}^{(2)}$  & $\mathbb{I}_{2,\mathcal{M}}^{(3)}$& $\mathbb{I}_{2,\mathcal{M}^{2}}^{(3)}$ \\ \hline
	$\{I,Z,X\}$ & 14.06 & 3 & 27  & 27 \\ \hline
	$\{I,Z,X^{2}\}$ & 48  & 5 & 27  & 39  \\ \hline
    $\{I,Z,X^{3}\}$& $81.94$ & 3 & 27  & 27  \\ \hline
    $\{I,Z,X^{4}\}$& 96  & 5 & 39  & 125  \\ \hline
    $\{I,Z^{2},X^{2}\}$& 96  & 9 & 27  & 63  \\ \hline
    $\{I,Z^{2},X^{4}\}$& 0 & 9 & 39  & 205  \\ \hline
    $\{I,Z,X^{4}Z^{2}\}$ & 96  & 5 & 27  & 55  \\ \hline
    $\{I,Z,Z^{2}\}$& 0 & 5 & 35  & 55 \\ \hline
    $\{I,Z,Z^{3}\}$&0 & 5 & 31  & 55  \\ \hline
    $\{I,Z,Z^{4}\}$ &0 &5 & 39  & 125 \\ \hline
    $\{I,Z^{2},Z^{4}\}$ & 0  & 9 & 55 & 205  \\ \hline\hline
\end{tabular}
\label{tableCD8}
\end{center}
\end{table}

We put the detailed proof of Theorem \ref{thm4} into Appendix \ref{appendix2}. In the following we give a rough analysis referring to the classification of $p^{2}\otimes p^{2}$ system. The preliminary of classification is similar to $p^{2}\otimes p^{2}$ system. Classes \textcircled{1};\textcircled{2};\textcircled{3};\textcircled{4} is somewhat like \textcircled{4};\textcircled{1};\textcircled{3},\textcircled{5};\textcircled{2} of Theorem \ref{thm3} respectively, while the analysis of classes \textcircled{4} is more intricate than \textcircled{2} of Theorem \ref{thm3}. More precisely, there are no corresponding cases in Theorem \ref{thm3} for \textcircled{4}. The appearance of case \textcircled{5} is aroused by the fact
\begin{small}
$$\mathcal{N}_{(k,s,t_1,\hat{t}_{2})}=\{I_{p^{s}}\otimes I_{p^{\alpha-s}},I_{p^{s}}\otimes Z_{p^{\alpha - s}}, X_{p^{s}}^{kp^{s+t-\alpha}}\otimes Z_{p^{\alpha-s}}^{\hat{t}_2}\}$$
\end{small}
where $\hat{t}_{2}$ is the minimal element of $[[ t_2]]_{p^{t_1-s}}, s+t\geq \alpha$. Thus there are less inequivalent classes in this case than case \textcircled{4}. We let $1\leq k\leq \frac{p^{\alpha -s-t_1}}{2}$ for $\mathcal{M}_{(k,s,t)}$ in case \textcircled{2} rather than $1\leq k \leq \frac{p^{\alpha - t}}{2}$ since $\mathcal{M}_{k,s,t}\sim\mathcal{M}_{k',s,t}$ for $k\equiv k'\text{ (mod }p^{\alpha - s - t})$ via Lemma \ref{permutation} by a little conversion, so as in class \textcircled{4}.

%
%


\section{\label{sec:level6}summary}
In summary, we find some more general unitary operators besides Clifford operators, which are indispensable for our classification. Meanwhile, we present a new method to prove that two GPM-triples are U-inequivalent. Based on these unitary operators and necessary conditions, we successfully classify the sets of GPMs in $d\otimes d$ quantum systems, including pairs in any dimensional quantum system and triples in a power of prime quantum system.

We find out in $d\otimes d$ quantum system, the U-inequivalent GPM-pairs are equal to the number of the factors of $d$ for any $d$, while the U-inequivalent GPM-triples are about $\frac{(\alpha + 3)}{6}p^{\alpha} + O(\alpha p^{\alpha-1})$ when $d$ is a power of prime. By the process of our classification, we conjecture that the U-inequivalent GPM-triples are polynomial of the product of $d$ and the factors of $d$ for general $d$. Furthermore, we wish that our results will provide some new thoughts for classification of quadruples or more in general quantum system.

%

\appendix

\section{\label{appendix3}Two permutations \eqref{mode-4} and \eqref{mode-6} of Theorem \ref{thm-inequivalence}}
The followings are two possible permutations of $\mathcal{N}_{k,t}$.

(a) $U_{1}\langle I,Z^{p^s},X^{kp^{t}}Z^{t'p^{s}}\rangle U_{2}\approx\langle X^{-kp^{t}}Z^{t'p^{s}},Z^{p^s},I\rangle$, where $k\perp p,s<t$, then we have
$$ Z^{p^s}\widesim{U_1} X^{kp^{t}}Z^{(1-t')p^{s}} \text{ and }X^{p^{t}}\widesim{U_1} X^{(1-t')p^{t}}Z^{t'(t'-2)p^{s}k^{-1}}.$$
If the above transformation holds, then $t'\equiv 2 \text{ (mod } p^{t-s})$.

 (b) $U_{1}\langle I,Z^{p^s},X^{kp^{t}}Z^{t'p^{s}}\rangle U_{2}\approx\langle I,X^{-kp^{t}}Z^{t'p^{s}},Z^{p^s}\rangle$, where $k\perp p,s<t$, then we have
$$ Z^{p^{s}}\widesim{U_{1}} X^{-kp^{t}}Z^{t'p^{s}}\text{ and } X^{p^{t}}\widesim{U_{1}} X^{t'p^{t}}Z^{(1-{t'}^{2})p^{s}k^{-1}}.$$

If the above transformation holds, then $t'\equiv p^{t-s}-1 (\text{mod } p^{t-s})$.

\section{\label{appendix1}Classification of GPM-triples in $p^2\otimes p^{2}$ quantum system}

In the following, we prove the GPM triples are inequivalent when the parameters are distinct for the inside of cases \textcircled{1},\textcircled{3},\textcircled{4} of Theorem \ref{thm3}.

\textbf{Cases \textcircled{1}}
Since $\mathbb{I}^{(1)}_{\mathcal{M}_{s}}=48(1-\cos{2s\pi/p^{2}})(1\leq s\leq\lfloor\frac{p^{2}}{2}\rfloor)$, which is distinct for each $s$ since $1-\cos x$ is monotonous in $(0,\pi)$ and greater than 0 strictly.


%

\textbf{Case \textcircled{3}}
By Tian \emph{et al.} ~\cite{Tian16}
Since $\mathcal{N}_{s}$ and $\mathcal{N}_{s'}$ are locally inequivalent if $s'\not\in[s]$ by ~\cite{Tian16} in prime dimension, which also holds when we adding $[kp]$ into it, since
 $$\mathbb{I}_{\mathcal{N}_{s'},s}^{(2)}< \mathbb{I}_{\mathcal{N}_{s},s}^{(2)},$$
where $s'$ is the element of $[kp]$.
Since $[kp]=\{kp,1-kp,(p-k)p,1-(p-k)p\}$ for $1\leq k \leq \frac{p-1}{2}$ and $[2]=\{2,\frac{p+1}{2},p-1\}$, let $A=[kp]\cup [2]$ and when $p=6k+5$, the rest equivalence class have exactly $6$ elements (by Lemma \ref{LemInv}), thus when $p=6k+5$,
$$p^{2}=6\sum_{[a]\not\in A}|[a]|+5+2(p-1)$$
Then there can be classified into $\frac{p^{2}-2p-3}{6}+\frac{p-1}{2}+1$ local equivalence classes.

When $p=6k+1(p\geq 5)$ there exists a pair $\{a_{0},1-a_{0}\}$ (by Lemma \ref{LemInv}), thus the number of inequivalent classes equals $\frac{p^{2}-2p+1}{6}+\frac{p-1}{2}+1$.

We can get the following result by combing them together,
$$\lfloor\frac{p^{2}-2p+1}{6}\rfloor+\lfloor\frac{p-1}{2}\rfloor+1$$
When $p=3$, there are two equivalence classes $[2]=\{2,5,8\}, $ and $[4]=\{3,4,6,7\}$, it also satisfies the above equality. When $p=2$, there is only one equivalence class $[2]=\{2,3\}$, the above equality also holds in this case.

\textbf{Case \textcircled{5}}
 Since
$$\mathbb{I}_{\mathcal{W}_{\bar{s}},s}^{(3)}\ne\mathbb{I}_{\mathcal{W}_{s'},s}^{(3)}$$
 where $\bar{s}$ is the element of $[s]_{p}$ and $s'\not\in [s]_{p}$. Thus the classification is minimum. We can shrink the domain from $\{1,\cdots,p^{2}\}$ into $\{1,\cdots,p\}$ and don't change the value of $p[k]_{p}$, thus there are eventually $\lfloor\frac{p}{6}\rfloor+1$ local inequivalence classes in total by \cite{Tian16}.

\section{\label{appendix2}The classification of GPM-triples in $p^{\alpha}\otimes p^{\alpha}$ quantum system}

In the following, we prove the classification of GPM-triples in $p^{\alpha}\otimes p^{\alpha}$ quantum system in Theorem \ref{thm4} are proper and complete.

\begin{proof}
Since $p^{s}(1\leq s<\alpha)$ is non-invertible in $\mathbb{Z}_{p^{\alpha}}$. We can classify all of GPM-triples into the following 3 cases primarily:

(a) $\{I,Z^{p^{s}},X^{kp^{t}}\}, \text{ where } 0\leq s,t < \alpha, k\perp p$.

(b) $\{I,Z^{p^{s}},X^{kp^{t_{1}}}Z^{k'p^{t_{2}}}\}, \text{ where } 0\leq s,t_{1},t_2\leq\alpha,$ and $k,k'\perp p $.

(c) $\{I,Z^{p^{s}},Z^{kp^{t}}\}, \text{ where } 0\leq s,t<\alpha \text{ and }k\perp p$.

We classify the above equivalence classes with the following equivalent relations, and then prove our classifications are minimum.

\begin{widetext}
\begin{flalign}
&\{I,Z^{p^{s}},X^{kp^{t}}\}\sim\{I,Z^{p^{t}},X^{-kp^{s}}\}\label{eq-a1}\\
&\{I,Z^{p^{s}},X^{kp^{t}}\}\mathop\sim\limits^{Q_{-1}}\{I,Z^{-p^{s}},X^{-kp^{t}}\}\sim\{I,Z^{-p^{s}},Z^{p^{s}}X^{-kp^{t}}\}\sim\{I,Z^{p^{s}},X^{-kp^{t}}\},\text{ where } s\geq t\label{eq-a2}\\
&\{I,Z^{p^{s}},X^{p^{t}}\}=\{I_{p^{s}}\otimes I_{p^{\alpha - s}},I_{p^{s}}\otimes Z_{p^{\alpha - s}},X_{p^{s}}^{p^{s+t-\alpha}}\otimes I_{p^{\alpha-s}}\}\sim \{I,Z^{p^{s}},X^{kp^{t}}\}, \text{ where } s+t\geq \alpha.\label{eq-a3}\\
&\{I,Z^{p^{s}},X^{kp^{t}}Z^{t'p^{s}}\}\sim \begin{cases}
&\{I,Z^{p^{s}},X^{-kp^{t}}Z^{(1-t')p^{s}}\}\\
&\{I,Z^{p^{s}},X^{-kp^{t}}Z^{t'^{-1}p^{s}}\}\text{, where } t'\perp p.\\
&\{I,Z^{p^{s}},X^{kp^{t}}Z^{(1-t')^{-1}p^{s}}\}\text {, where } 1-t \perp p.
\end{cases}\label{eq-c1}\\
&\{I,Z^{p^{s}},X^{t}Z^{t'}\}\sim \{I,Z^{p^{s}},X^{t}\}\sim\{I,Z,X^{-p^{s}t^{-1}}\}
\text{, where } t\perp p\text{ or } t|p \text{ and } t'|p.\label{eq-c2}\\
&\{I,Z^{p^{s}},X^{kp^{t}}Z^{k'p^{t'}}\}\sim\{I,Z^{p^{t'}},X^{-kp^{(t-t'+s)}}Z^{k'^{-1}p^{s}}\} \text{, where } t'<t.\label{eq-c3}\\
&\{I,Z^{p^{s}},X^{kp^{t}}Z^{k'p^{t'}}\}\widesim{P^{k'k^{-1}p^{t'-t}}} \{I,Z^{p^{s}},X^{kp^{t}}\},\text{ where }t'\geq t.\label{eq-c4}\\
&\{I,Z^{p^{s}},X^{kp^{t}}Z^{p^{s}}\}\mathop\sim\limits^{Q_{-1}}\{I,Z^{-p^{s}},X^{-kp^{t}}Z^{-p^{s}}\}\sim \{I,Z^{p^{s}},X^{kp^{t}}\}, \label{eq-c5}\\
&\{I,Z^{p^{s}},X^{kp^{t}}Z^{p^st_2}\}=\{I_{p^{s}}\otimes I_{p^{\alpha-s}},I_{p^{s}}\otimes Z_{p^{\alpha - s}}, X_{p^{s}}^{kp^{s+t-\alpha}}\otimes Z_{p^{\alpha-s}}^{t_2}\}\sim\{I,Z^{p^{s}},X^{p^{t}}Z^{p^st_2}\}, \text{where } s+t\geq \alpha \label{eq-c6}\\
&\{I,Z^{p^{s}},Z^{kp^{t}}\}\sim\{I,Z^{p^{t}},Z^{k^{-1}p^{s}}\}\label{eq-b1}\\
&\{I,Z^{p^{s}},Z^{p^{s}t}\}\sim\begin{cases}
&\{I,Z^{p^{s}},Z^{p^{s}(1-t')^{-1}}\},\text{ where }(1-t)\perp p\\
&\{I,Z^{p^{s}},Z^{p^{s}(1-t^{-1}})\},\text{ where } t\perp p.
\end{cases}\label{eq-b2}
\end{flalign}
\end{widetext}

By Equation \eqref{eq-a1}, we can restrict $0\leq s\leq t<\alpha$ in case (a). Now we split case (a) into two parts, $s+t\geq \alpha$ and $s+t<\alpha$. When $s+t\geq \alpha, \{I,Z^{p^{s}},X^{kp^{t}}\}$ are equivalent for different $k$ by Equation \eqref{eq-a3}. Thus we get equivalence classes \textcircled{1} in Theorem \ref{thm4}. Now we prove the classes in \textcircled{1} are minimum.

$\mathcal{M}_{{(s,t)}}$ are U-inequivalent for different $s,t$, since $\mathbb{I}_{\mathcal{M}_{{(s,t)}},\alpha - s}^{(2)}=9$, while $\mathbb{I}_{\mathcal{M}_{(s',t')},\alpha - s}^{(2)}\leq 5$ for $s>s'$. Meanwhile, $\mathbb{I}_{\mathcal{M}_{{(s,t)}},\alpha - t}^{(2)}=5$, while $\mathbb{I}_{\mathcal{M}_{(s,t')},\alpha - t}^{(2)}=3$ for $t>t'$. Therefore, $\mathcal{M}_{(s,t)}$ are U-inequivalent for different $s,t$.

Since whether $\alpha$ is odd or even has little influence to the order of quantity of U-inequivalent classes, thus we suppose $\alpha$ is even. Thus there are exactly $\frac{\alpha^{2}}{4}$ U-inequivalent classes in \textcircled{1} of Theorem \ref{thm4}.

When $s+t<\alpha$, we can additional restrict $k\leq \frac{p^{\alpha-t}}{2}$ for $\mathcal{M}_{(k,s,t)}$ by equivalence transformations \eqref{eq-a1},\eqref{eq-a2}, thus we get classes $\mathcal{M}_{(k,s,t)}=\{I,Z^{p^{s}},X^{kp^{t}}\}$, where $0\leq s\leq t<\alpha, s+t<\alpha, 1\leq k\leq \frac{p^{\alpha -t}}{2}$ and $k\perp p$.

In the following, we prove $\mathcal{M}_{(k,s,t)}\sim \mathcal{M}_{(k',s,t)}$ for $k\equiv k'\text{ (mod }p^{\alpha-s-t})$.

Since $X^{kp^{t}}\sim X^{k'p^{t}}$ for $k,k'$ are invertible if $k\equiv k'\text{ (mod }p^{\alpha - s - t})$ and $X^{p^{t}}\sim X^{k_1p^{\alpha - s}+p^{t}}$, where $k_1$ satisfies $kk_1p^{\alpha - s - t}=k-k'$. Then by Lemma \ref{conjugate transformation}, $\mathcal{M}_{(k,s,t)}\sim \mathcal{M}_{(k',s,t)}$ for $k\equiv k'\text{ (mod }p^{\alpha-s-t})$. Thus we get classes \textcircled{2} of Theorem \ref{thm4}.

In the following we prove the classes in \textcircled{2} are minimal. Consider two GPM-triples $\mathcal{M}_{(k,s,t)}$ and $\mathcal{M}_{(k',s',t')}$. We find that $\mathcal{M}_{(k,s,t)}$ are U-inequivalent for different $s,t$ with analysis similar to the above case $s+t\geq \alpha$. For different $k$, since $\mathbb{I}_{\mathcal{M}_{(k,s,t)}}^{(1)}=48(1-\cos\frac{2kp^{(s+t)}\pi}{p^{\alpha}})$.
If $s=s',t = t'$ and $k\not\equiv k'\text{ (mod }p^{\alpha-s-t})$, their $\mathbb{I}^{(1)}$ are distinct.

There are exactly $\frac{2p^{\alpha + 3} - \alpha p^{3} - 2p^{3} + \alpha p}{4(p-1)^{2}(p+1)}$ U-inequivalent classes in \textcircled{2} of Theorem \ref{thm4} when $\alpha$ is even.

In case (b), by equivalent transformations \eqref{eq-c2},\eqref{eq-c3},\eqref{eq-c4}, we can restrict $t_{2}<t_{1},s<t_{1}$ for case (b). Let $k'p^{t_2}=p^{s}t'$ for $t'=k'p^{t_2-s}$, by equivalence transformation \eqref{eq-c5}, we can restrict $t_2>1$, together with equivalence transformation \eqref{eq-c1}, we can restrict case (b) into classes $\mathcal{N}_{(k,s,t_1,\hat{t}_2)}$, where $\hat{t}_2$ is the minimal element of $[[ t_2]]_{p^{t_1-s}}$. Fianlly, we get classes
$$\mathcal{N}_{(k,s,t_1,\hat{t}_2)}$$
 where $0\leq s<t_1<\alpha, s+t_1<\alpha,1\leq k\leq \frac{p^{\alpha - t_1}}{2}, 2\leq t_2<p^{t_1-s}, k\perp p$, $\hat{t}_2$ is the minimal element of $[[ t_2]]_{p^{t_1-s}}$, and classes
 $$\mathcal{N}_{(s,t_1,\hat{t}_2')}$$
 where $0\leq s<t_1<\alpha,s+t_1\geq \alpha, 2\leq t_2<p^{t_1 -s}$, and $\hat{t}_2'$ is the minimal element of $[t_2]_{p^{t_1-s}}$ by equivalent transformation \eqref{eq-c6} and some obvious analysis. Meanwhile we can restrict $k$ to be between 1 and $\frac{p^{\alpha -s- t_1}}{2}$ for classes $\mathcal{N}_{(k,s,t_1,\hat{t}_2)}$ where $\hat{t}_2$ is the minimal element of $[[ t_2]]_{p^{t_1-s}}$ by Lemma \ref{conjugate transformation}, thus obtain case \textcircled{4} of Theorem \ref{thm4}.

We show that $\mathcal{N}_{(k,s,t_1,\hat{t}_2)}$ are U-inequivalent for different $s,t_{1}$, where $\hat{t}_2$ is the minimal element of $[[ t_2]]_{p^{t_1-s}}$, $s,t<\alpha$ and $k\perp p$.

\begin{itemize}
\item $\mathbb{I}_{\mathcal{N}_{(k,s,t,t_2)},\alpha - s}^{(2)}>\mathbb{I}_{\mathcal{N}_{(k',s',t',t'_2)},\alpha - s}^{(2)}$ for $s>s'$.
\item  $\mathbb{I}_{t_2,\mathcal{N}_{(k,s,t,t_2)}^{\alpha - t}}^{(3)}\ne \mathbb{I}_{t_2,\mathcal{N}_{(k',s,t',t'_2)}^{\alpha - t}}^{(3)}$ for $t>t'$.
\item $\mathbb{I}_{\mathcal{N}_{(k,s,t,t_2)}}^{(1)}=48(1-\cos\frac{2kp^{(s+t)}\pi}{p^{\alpha}})$, thus $\mathcal{N}_{(k,s,t,t_2)}\not\sim \mathcal{N}_{(k',s,t,t'_2)}$ for different $k,k'\in p^{\alpha -s -t}$ and $s+t<\alpha$.
\item $\mathbb{I}_{t_2,\mathcal{N}_{(k,s,t,\bar{t}_2)}^{\alpha - t}}^{(3)}\ne \mathbb{I}_{t_2,\mathcal{N}_{(k,s,t,t')}^{\alpha - t}}^{(3)}$ for $t'\not\in[t_2]_{p^{t_1-s}}$ and $\bar{t}_2$ is the element of $[t_2]_{p^{t_1-s}}$.
\end{itemize}
Furthermore, by Theorem \ref{thm-inequivalence}, $\mathcal{N}_{(k,s,t,t_2)}\not\sim \mathcal{N}_{(-k,s,t,t_2)}$ when $t_2\not\in [2]_{p^{t_1-s}}$, and $p\ne 2$ or $s+t+1\ne\alpha$. Thus \textcircled{4} and \textcircled{5} are minimum. When $p=2$ and $s+t+1=\alpha$, $\mathcal{N}_{(k,s,t,t_2)}\sim\mathcal{N}_{(-k,s,t,t_2)}$ by Lemma \ref{conjugate transformation} since $k\equiv -k\text{ (mod }p^{\alpha - s - t})$.

There are exactly
$$\frac{(2\alpha - 2)p^{\alpha + 2} - (2+2\alpha)p^{\alpha} - 6p^{\alpha +1}+(3\alpha + 2)p^{2} + 6p - 3\alpha +2}{12(p^{2}-1)}$$
 U-inequivalent classes in \textcircled{4},
 $$\frac{2p^{\alpha}-\alpha p^{2}  +\alpha -2}{12(p-1)^{2}}$$
  U-inequivalent classes in \textcircled{5} of Theorem \ref{thm4} when $p=6k+5$ or $p = 3$, and
  $$\frac{(2\alpha -2)p^{\alpha + 2} + 2p^{\alpha + 1} -(2+2\alpha)p^{\alpha} +(2-\alpha)p^{2}-2p + \alpha +2}{12(p^{2}-1)}$$
 U-equivalent classes in \textcircled{4} ,
 $$\frac{2p^{\alpha} -\alpha p^{2} + \alpha -2}{12(p-1)^{2}} + \frac{\alpha(\alpha -2)}{6}$$
  U-inequivalent classes in \textcircled{5} of Theorem \ref{thm4} if $p=6k+1$.

  In the other hand, when $p= 2$. There are two parts of \textcircled{4}, and the U-inequivalent classes are respectively:
  $$\frac{(3\alpha - 8)2^{\alpha} - 18\alpha + 44}{18},$$
  and
  $$\frac{2^{\alpha} + 3\alpha -10}{6},$$
  The U-inequivalent classes in \textcircled{5} are
  $$\frac{2^{\alpha + 1} + 3\alpha^{2} - 24\alpha + 28}{12}$$

 By equivalent transformation \eqref{eq-b1} we can restrict $s<t$ of case (c). Let $kp^{t}=t'p^{s}$ for $t'=kp^{t-s}$, then we get case \textcircled{3} in Theorem \ref{thm4} by classification of case (c) via equivalent transformation \eqref{eq-b2}.

It's easy to find that for different $s,s', \mathcal{N}_{(s,t)}\not\sim \mathcal{N}_{(s',t')}$. On the other hand, $\mathbb{I}_{t,\mathcal{N}_{(s,\bar{t})}}^{(3)}\ne\mathbb{I}_{t,\mathcal{N}_{(s,t')}}^{(3)}$ for $t'\not\in [t]_{p^{\alpha-s}}$, where $\bar{t}$ is the element of $[t]_{p^{\alpha-s}}$, thus \textcircled{3} are minimum.

There are exactly $\frac{(p+1)(p^{\alpha }-1)}{6(p-1)}$ U-inequivalent classes when $p=6k+5$ or $p = 3$, and $\frac{(p+1)(p^{\alpha }-1)}{6(p-1)} + \frac{2\alpha}{3}$ U-equivalent classes when $p=6k+1$ in \textcircled{3} of Theorem \ref{thm4} (The quantity is equal to $2^{\alpha - 1} + \alpha -3$ when $p = 2$). Thus, there are $\frac{(\alpha + 3)}{6}p^{\alpha} + O(\alpha p^{\alpha-1})$ U-inequivalent classes in summation when $p\geq 3$. In the other hand, when $p = 2$, the U-inequivalent classes equals
$$\frac{(3\alpha + 19)2^{\alpha}}{18} + \frac{\alpha^{2}}{2} -\frac{7\alpha}{4}-\frac{5}{9}$$
for even $\alpha$ and $\alpha>2$. Thus it is easy to check the quantity $\frac{(\alpha + 3)}{6}p^{\alpha} + O(\alpha p^{\alpha-1})$ also holds for $p = 2$.

In the following, we prove that all the cases of Theorem \ref{thm4} are inequivalent.

\begin{itemize}
\item \textcircled{1} and \textcircled{2} are unitary inequivalent, since all of triples in \textcircled{1} are commute respectively, while all of triples in \textcircled{2} are non-commute respectively.
\item \textcircled{1},\textcircled{2} and \textcircled{4},\textcircled{5} are unitary inequivalent, since
$$\mathbb{I}_{\mathcal{M}_{(k,s,t)},p^{\alpha -s}}^{(2)}>\mathbb{I}_{\mathcal{N}_{(k',s',t',t_2)},p^{\alpha - s}}$$
where $s>s'$ and when $s<s' $ or $t\geq t',\mathcal{M}_{(k,s,t)}\not\sim\mathcal{N}_{(k',s',t',t_2)}$ in the same way. On the other hand,
$$\mathbb{I}_{t_2,\mathcal{M}_{(k,s,t)}^{p^{\alpha - t'}}}^{(3)}\ne\mathbb{I}_{t_2,\mathcal{N}_{(k',s,t,t_2)}^{p^{\alpha - t'}}}^{(3)}$$
where $t<t'$.
\item  \textcircled{3} are unitary inequivalent to \textcircled{1},\textcircled{2} since
$$\mathbb{I}_{t,\mathcal{N}_{(s,\bar{t})}}^{(3)}\ne\mathbb{I}_{t,\mathcal{M}_{k,s',t'}}^{(3)}$$
where $\bar{t}$ is the element of $[[ t]]_{p^{\alpha - s}}$, meanwhile, \textcircled{3} are unitary inequivalent to \textcircled{4},\textcircled{5} in the same way.
\item The classifications of \textcircled{4} and \textcircled{5} also imply that \textcircled{4} and \textcircled{5} are unitary inequivalent.

\end{itemize}
Thus the classification of Theorem \ref{thm4} are complete.
\end{proof}

\bibliographystyle{apsrev4-1}
\bibliography{bibtex}

\end{document}